\documentclass{article}
\usepackage{graphicx}
\usepackage{amsfonts}
\usepackage{amsmath}
\usepackage{amssymb}
\usepackage{natbib}
\usepackage{mdframed}
\usepackage{tikz}
\usetikzlibrary{arrows.meta, quotes, calc, decorations.pathreplacing, fit, shapes, positioning, shadows, shapes.geometric}
\usepackage{pgfplots}
\pgfplotsset{compat=newest}

\newtheorem{theorem}{Theorem}
\newtheorem{lemma}{Lemma}

\newtheorem{proof}{Proof}

\newcommand\blfootnote[1]{%
  \begingroup
  \renewcommand\thefootnote{}\footnote{#1}%
  \addtocounter{footnote}{-1}%
  \endgroup
}

\title{Differential Equation Approximations for Population Games using Elementary Probability}
\author{Semih Kara, Nuno C. Martins}

\begin{document}

\maketitle
\blfootnote{This work was supported by AFOSR FA95502310467 and NSF ECCS 2139713, NSF CNS 213556.}
\blfootnote{The authors are with the Electrical and Computer Engineering Department, and the Institute for Systems Research, University of Maryland, College Park, MD 20742, USA.{\tt\small \{skara, nmartins\}@umd.edu.}}

\begin{abstract}
    Population games model the evolution of strategic interactions among a large number of uniform agents. Due to the agents' uniformity and quantity, their aggregate strategic choices can be approximated by the solutions of a class of ordinary differential equations. This mean-field approach has found to be an effective tool of analysis. However its current proofs rely on advanced mathematical techniques, making them less accessible. In this article, we present a simpler derivation, using only undergraduate-level probability.

    
    
\end{abstract}

\section{Introduction}
\label{sec:Intro}

Population games is a framework for dynamic strategic interactions of many agents; finding applications in diverse fields such as traffic management \citep{Smith1984The-stability-o}, electricity demand regulation \citep{Resilient_Distributed_Real-Time_Demand_Response_via_Population_Games_Srikantha_Kundur,Quijano2017The-role-of-pop}, distributed task allocation \citep{Multi-Robot_Task_Allocation_Games_in_Dynamically_Changing_Environments_Park_et_al,Population_Games_With_Erlang_Clocks_Convergence_to_Nash_Equilibria_For_Pairwise_Comparison_Dynamics_Kara_Martins_Arcak}, distributed extremum seeking \citep{Shahshahani_gradient-like_extremum_seeking_Poveda_Quijano,Distributed_Population_Dynamics_Optimization_and_Control_Applications_Barreiro-Gomez_et_al}, communication networks \citep{Evolutionary_dynamics_and_potential_games_in_non-cooperative_routing_Altman_et_al}, and more \citep{Sandholm2010Population-Game,Park2018Payoff-Dynamic-}.

In this framework, there is a large number of $N$-many uniform agents, each choosing a strategy from a set of $n$ options. Because that the agents are uniform, an important quantity is the fractions of the population following each strategy, denoted as $X^N=(X^N_1,\dots,X^N_n)$ and referred to as the population state. Every strategy has a payoff that quantifies the utility that an agent would receive upon choosing it.

The agents repeatedly revise their strategies by taking $X^N$ and the payoffs into account. Each agent has its own independent and identical Poisson clock, which dictates when it revises its strategy. When a revision occurs, the agent uses a learning rule to modify its strategy, causing $X^N$ to vary over time. The resulting $X^N$ is a right-continuous pure-jump Markov process. We give a more detailed description of the framework in Section~\ref{sec:Pop_Games_with_Fin_Many_Agents}.

\subsection{Population Games and Mean-Field Approximation}
\label{subsec:Pop_Games_and_MFA}

Many research on population games, including the references above, rely on the mean-field approximation of $X^N$. In population games, ``mean-field'' signifies aggregating the effects of individual agents. The key idea is that the agents' homogeneity cause the randomness that they introduce to the aggregate behavior to balance out as $N\to\infty$. Hence, the population's aggregate strategic behavior, captured by $X^N$, is well-approximated by a deterministic trajectory. The figure below illustrates this idea.
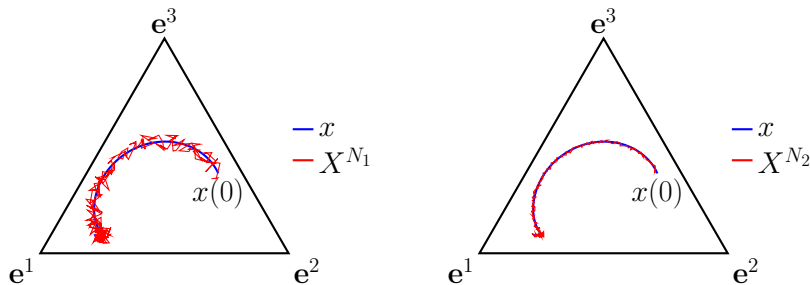
\begin{figure}[ht]
\centering
\begin{tikzpicture}[scale=0.55,transform shape]
    \def\a{6} 
    \def\T{100} 
    \def\f{0.5} 
    \def\rZero{\a/12} 
    \def\xOffset{0.05 * \a} 
    \def\yOffset{0.75 * sqrt(3)*\a/6} 
    \coordinate (A) at (0,0);
    \coordinate (B) at (\a,0);
    \coordinate (C) at (\a/2, {sqrt(3)*\a/2});
    \coordinate (Center) at ($1/3*(A)+1/3*(B)+1/3*(C)$);
    \draw[thick] (A) -- (B) -- (C) -- cycle;
    \node[below left, font=\huge] at (A) {$\mathbf{e}^1$};
    \node[below right, font=\huge] at (B) {$\mathbf{e}^2$};
    \node[above, font=\huge] at (C) {$\mathbf{e}^3$};
    \coordinate (X0) at ({(\a/4) * (1 + cos(2*pi*\f*0)) + \a/6 + \xOffset},
                          {(\a/4) * (1 + sin(2*pi*\f*0)) + sqrt(3)*\a/6 - \yOffset});
    \node[below, font=\huge] at (X0) {$x(0)$};
    \draw[blue, thick] plot[variable=\t, domain=0:\T, samples=200, smooth] 
        ({(\a/4) * (1 - \t/\T) * (1 + cos(2*pi*\f*\t)) + \a/6 + \xOffset},
         {(\a/4) * (1 - \t/\T) * (1 + sin(2*pi*\f*\t)) + sqrt(3)*\a/6 - \yOffset});
    \draw[red] plot[variable=\t, domain=0:\T, samples=200] 
        ({(\a/4) * (1 - \t/\T) * (1 + cos(2*pi*\f*\t)) + \a/6 + \xOffset + 0.2*rand},
         {(\a/4) * (1 - \t/\T) * (1 + sin(2*pi*\f*\t)) + sqrt(3)*\a/6 - \yOffset + 0.2*rand});
    \coordinate (Legend) at (\a, \a/2);
    \draw[blue, thick] ($(Legend)+(0.1,0)$) -- ++(0.5,0) node[right, black, font=\huge] {$x$};
    \draw[red, thick] ($(Legend)+(0.1,-0.75)$) -- ++(0.5,0) node[right, black, font=\huge] {$X^{N_1}$};
\end{tikzpicture}\hspace{0.3in}
\begin{tikzpicture}[scale=0.55,transform shape]
    \def\a{6} 
    \def\T{100} 
    \def\f{0.5} 
    \def\rZero{\a/12} 
    \def\xOffset{0.05 * \a} 
    \def\yOffset{0.75 * sqrt(3)*\a/6} 
    \coordinate (A) at (0,0);
    \coordinate (B) at (\a,0);
    \coordinate (C) at (\a/2, {sqrt(3)*\a/2});
    \coordinate (Center) at ($1/3*(A)+1/3*(B)+1/3*(C)$);
    \draw[thick] (A) -- (B) -- (C) -- cycle;
    \node[below left, font=\huge] at (A) {$\mathbf{e}^1$};
    \node[below right, font=\huge] at (B) {$\mathbf{e}^2$};
    \node[above, font=\huge] at (C) {$\mathbf{e}^3$};
    \coordinate (X0) at ({(\a/4) * (1 + cos(2*pi*\f*0)) + \a/6 + \xOffset},
                          {(\a/4) * (1 + sin(2*pi*\f*0)) + sqrt(3)*\a/6 - \yOffset});
    \node[below, font=\huge] at (X0) {$x(0)$};
    \draw[blue, thick] plot[variable=\t, domain=0:\T, samples=200, smooth] 
        ({(\a/4) * (1 - \t/\T) * (1 + cos(2*pi*\f*\t)) + \a/6 + \xOffset},
         {(\a/4) * (1 - \t/\T) * (1 + sin(2*pi*\f*\t)) + sqrt(3)*\a/6 - \yOffset});
    \draw[red] plot[variable=\t, domain=0:\T, samples=200] 
        ({(\a/4) * (1 - \t/\T) * (1 + cos(2*pi*\f*\t)) + \a/6 + \xOffset + 0.05*rand},
         {(\a/4) * (1 - \t/\T) * (1 + sin(2*pi*\f*\t)) + sqrt(3)*\a/6 - \yOffset + 0.05*rand});
    \coordinate (Legend) at (\a, \a/2);
    \draw[blue, thick] ($(Legend)+(0.1,0)$) -- ++(0.5,0) node[right, black, font=\huge] {$x$};
    \draw[red, thick] ($(Legend)+(0.1,-0.75)$) -- ++(0.5,0) node[right, black, font=\huge] {$X^{N_2}$};
\end{tikzpicture}
\caption{Effect of increasing $N$ on the approximation: $N_2 > N_1$, yielding a trajectory closer to the deterministic approximation.}
\end{figure}

More rigorously, this approximation has the following form \citep{Sandholm2003Evolution-and-e,Benaim2003Deterministic-a}: Given any $\epsilon,T>0$, if $\lim_{N\to\infty}X^N(0)=\mathrm{x}_0$, then
\begin{align}
\lim_{N\to\infty}\mathbb{P}\left( \sup_{t\in [0,T]} \|X^N(t)-x(t)\|_2 >\epsilon \right)=0 \label{eq:lim_to_0}
\end{align}
where $x$, with $x(0)=\mathrm{x}_0$, is the solution of an ordinary differential equation (ODE) characterized by $\mathbb{E}[X^N]$. There are also mean-field results that associate the stationary distributions of $X^N$ to the recurrent points of the ODE \citep{Sandholm2010Population-Game,Benaim1998Recursive-algor,Benaim2003Deterministic-a}. Nevertheless, in this paper, our focus is on approximations within a finite interval--in the sense of \eqref{eq:lim_to_0}.

\subsection{Towards a Simpler Proof}
\label{subsec:Need_for_a_Simpler_Proof}

Mean-field approximation for population games have a rich history. In the concluding notes of the chapter, \citet[Chapter~10]{Sandholm2010Population-Game} provides a concise overview of this history, which we reiterate here: The initial results were obtained under additional assumptions on the agents' behavior. \cite{Boylan1995Continuous-appr} demonstrated how evolutionary processes, driven by random matching, converge to deterministic trajectories as the population size grows. Following suit, \citet{Binmore1995Musical-chairs}, as well as \citet{Borges1997Learning-throug}, and \citet{Schlag_Why_Imitate_and_If_So_How_A_Boundedly_Rational_Approach_to_Multi-armed_Bandits}, studied evolutionary models that converge to the so-called replicator dynamics. \cite{Binmore_Samuelson_Evolutionary_Drift_and_Equilibrium_Selection} established a broader method of approximation, however their findings cover discrete-time processes that impose additional assumptions on the agents' revision times. A general result of type \eqref{eq:lim_to_0} was reported by \citet{Sandholm2003Evolution-and-e}, which is based on the seminal paper of \citet{Kurtz1970Solutions-of-Or}. Particularly, \citet[Theorem~2.1]{Kurtz1970Solutions-of-Or} showed that some classes of Markov processes can be approximated by solutions of ODEs and \citet{Sandholm2003Evolution-and-e} proved that $X^N$ belongs to this class, culminating in \eqref{eq:lim_to_0}. While Kurtz's results greatly extend beyond this application, their generality brings two setbacks: (i) They only guarantee that the approximation error approaches 0 in the limit, with no specific bound for the non-limit case. (ii) Their proofs involve advanced mathematical techniques, making them less accessible. The strongest known contribution is due to \citet{Benaim2003Deterministic-a}, which revealed an exponential bound (decaying with rate $\epsilon^2N$) on the probability in \eqref{eq:lim_to_0}. Unfortunately, they also employ advanced mathematical techniques and leave a significant portion of the derivations to the reader, including calculating the exact expression of the bound.

There is a recent body of research \citep{Park2018Payoff-Dynamic-,Kara2021Pairwise-Compar,Excess_Payoff_Evolutionary_Dynamics_With_Strategy-Dependent_Revision_Rates_Convergence_to_Nash_Equilibria_for_Potential_Games_Kara_Martins,Population_Games_With_Erlang_Clocks_Convergence_to_Nash_Equilibria_For_Pairwise_Comparison_Dynamics_Kara_Martins_Arcak} that has extended the population games framework and adjusted the mean-field approach of \citet{Sandholm2003Evolution-and-e} accordingly. Nonetheless, their approximation results essentially remain as applications of \citet[Theorem~2.1]{Kurtz1970Solutions-of-Or}. Note that, in this paper, we focus on the standard setting that we explain in Section~\ref{sec:Pop_Games_with_Fin_Many_Agents}, not these extensions.

Although \citet{Benaim2003Deterministic-a} and \citet{Kurtz1970Solutions-of-Or} (the latter through \citet{Sandholm2003Evolution-and-e}) prove that \eqref{eq:lim_to_0} holds, the concerns that we expressed above hinder their appeal to a general audience. In this paper, we adopt an educational stance to address these concerns: Our aim is to make the framework more accessible by giving an alternative proof of \eqref{eq:lim_to_0} and presenting an explicit bound for the non-limit case, using only undergraduate-level probability. We emphasize that we are \textit{not} strengthening the results of \citet{Benaim2003Deterministic-a} and \citet{Kurtz1970Solutions-of-Or}.

The remainder of this paper has two main parts: In Section~\ref{sec:Pop_Games_with_Fin_Many_Agents} we go over the population games framework and in Section~\ref{sec:Main_Res} we present our (simple) mean-field approximation result.

\section{Population Games with Finitely Many Agents}
\label{sec:Pop_Games_with_Fin_Many_Agents}

In this section, we present an overview of the population games and evolutionary dynamics framework. For the sake of discussion, we break the framework into two principal components:
\begin{itemize}
\item The strategic environment through which the agents interact (e.g. the set of available strategies, payoffs, and the indistinguishability of the agents).
\item When and how the agents decide on their strategies.
\end{itemize}

\subsection{The Strategic Environment}
\label{subsec:Strat_Env}

Consider a population of $N\in\mathbb{N}$ nondescript agents (meaning that the agents are indistinguishable). At any time $t\geq 0$, each agent follows a single strategy from the set $\{1,\dots,n\}$. For each strategy $i\in\{1,\dots,n\}$, we denote
\begin{align*}
X^N_i(t) := \frac{\text{Number of agents following strategy $i$ at time $t$}}{N},
\end{align*}
that is, $X^N_i(t)$ is the fraction of agents playing $i$ at time $t$. Moreover, we denote $X^N := (X^N_1,\dots,X^N_n)$, which is commonly referred to as the population state. Since the agents are nondescript, $X^N$ contains all the relevant information about the population's strategy profile. Observe that $X^N$ takes values in the discrete simplex $\mathbb{X}^N := \{ \mathrm{x}\in\mathbb{R}^n \ | \ N\mathrm{x}_i \in \{0, \ldots, N\}~\text{and}~\sum_{j=1}^n \mathrm{x}_j=1 \}$.

At any time $t\geq 0$, each strategy $i\in\{1,\dots,n\}$ has a payoff $P_i(t)$ that is assigned by a payoff mechanism. The payoff mechanism determines these payoffs by coupling the agents' strategies through $X^N$. Although the framework allows for more general payoff mechanisms \citep{Park2018Payoff-Dynamic-}, as the main objective of this article is tractability, we confine our focus to the memoryless case. So, we assume that a Lipschitz continuous function $\mathcal{F}:\Delta\to\mathbb{R}^n$ specifies the payoffs as $P(t):=(P_1(t),\dots,P_n(t))=\mathcal{F}(X^N(t))$, where $\Delta:=\{ \mathrm{x}\in\mathbb{R}^n \ | \ 0\leq \mathrm{x}_i\leq 1\}~\text{and}~\sum_{j=1}^n \mathrm{x}_j=1 \}$ is the $(n-1)$-dimensional simplex.

\subsection{The Revision Paradigm}
\label{subsec:Rev_Par}

The agents repeatedly revise their strategies in response to the strategic environment. In what follows, we characterize when and how these revisions happen.

\subsubsection{Revision Times}
\label{subsubsec:Rev_Times}

The revision times are characterized by Poisson processes: The agents possess independent and identical Poisson clocks (processes) with rate $\lambda$, where a ``tick'' (jump) in an agent's clock indicates a strategy revision by that agent. 

Consequently, the number of revisions performed by the entire population is the sum of the agents' clocks, which is a Poisson process with rate $\lambda N$. Given any $\bar{t}>\underline{t}\geq 0$, we denote
\begin{align*}
&\mathbb{I}([\underline{t},\bar{t}]) := \; \parbox{0.52\textwidth}{The number of revisions performed by the population over the time interval $[\underline{t},\bar{t}]$.}
\end{align*}
From the properties of Poisson processes, it follows for every $t,\delta>0$ that
\begin{align}
& \mathbb{P}( \mathbb{I}([t,t+\delta]) = 0 ) = e^{-N\lambda\delta}, \tag{$\mathbb{I}_0$}\\
& \mathbb{P}( \mathbb{I}([t,t+\delta]) = 1 ) = N\lambda\delta e^{-N\lambda\delta}, \tag{$\mathbb{I}_1$}\\
& \mathbb{P}( \mathbb{I}([t,t+\delta]) \geq 2 ) = 1-(1+N\lambda\delta)e^{-N\lambda\delta}. \tag{$\mathbb{I}_{\geq 2}$}
\end{align}
For notational convenience, we define
\begin{align*}
& \mathcal{M}^N(\delta) := 1-(1+N\lambda\delta)e^{-N\lambda\delta}.
\end{align*}
In the remainder of this paper, we frequently use ($\mathbb{I}_0$), ($\mathbb{I}_1$), ($\mathbb{I}_{\geq 2}$) and the fact that \begin{align}
\label{eq:MN_is_o(delta)}
\lim_{\delta\to 0^+} \frac{\mathcal{M}^N(\delta)}{\delta^2}=\frac{(N\lambda)^2}{2}.
\end{align}

\subsubsection{Strategy Switching Probabilities}
\label{subsubsec:Strat_Sw_Probs}

Having characterized the revision times, now we describe how a revising agent acts. Suppose that an agent's clock ticks at $t^*$. Then, by the properties of the Poisson process, there exists $\delta>0$ such that $\mathbb{P}( \mathbb{I}([t^*-\delta,t^*))=0 )=1$. That is, there is an epoch of length $\delta$ that precedes $t^*$ during which no revision occurs.

Let us denote the strategy (at $t^*-\delta$) of the revising agent by $i$. We assume that this agent chooses to play strategy $j$ with probability $\mathbb{P}(i\rightsquigarrow j \ | \ i)$, which can depend on $X^N(t^*-\delta)$ and $P(t^*-\delta)$. The realization of the agent's decision becomes its strategy at $t^*$. Thus, if the agent decides to play strategy $j$, then 
\begin{align*}
X^N(t^*) - X^N(t^*-\delta) = \frac{\mathbf{e}^j-\mathbf{e}^i}{N},
\end{align*}
where $\mathbf{e}^i$ denotes the $i$-th canonical basis vector of $\mathbb{R}^n$. 
    
Since the agents' clocks are independent and identical, the probability that an $i$-strategist performs the revision at $t^*$ is $X_i^N(t^*-\delta)$. Therefore, unconditional on the origin strategy of the revising agent, the revision results in a switch from $i$ to $j$ with probability $X_i^N(t^*-\delta)\mathbb{P}(i\rightsquigarrow j \ | \ i)$. Figure~\ref{fig:str_sw_timeline} illustrates the revision timeline near $t^*$.
\begin{figure}[ht]
\centering
    \begin{tikzpicture}[scale=0.9]
        \tikzstyle{every node}=[]
        \node (label) at (0,0.2\textheight) {\large $t$};
        \node (t_start) at (label.east) {};
        \node (t_end) at (0.97\textwidth,0.2\textheight) {};
        \draw[->] (t_start) -- (t_end);
        \node[right=0.3 of t_start] (rt1) {};
        \node[right=2.4 of rt1] (rt2) {};
        \node[right=2.1 of rt2] (rt3) {};
        \node[left=0.3 of t_end] (rt4) {};
        \draw[line width=2.5pt, color=green] (rt1)++(0,-0.2) -- ++(0,0.4);
        \draw[line width=2.5pt, color=green] (rt2)++(0,-0.2) -- ++(0,0.4);
        \draw[line width=2.5pt, color=green] (rt3)++(0,-0.2) -- ++(0,0.4);
        \draw[line width=2.5pt, color=green] (rt4)++(0,-0.2) -- ++(0,0.4);
        \node[left={0.65 of rt2}, below={0.65 of rt2}] (bn2) {};
        \node[left={0.65 of rt3}, below={0.65 of rt3}] (bn3) {};
        \node[above=8pt] at ($(rt1)!0.5!(rt2)$) {$\sim\exp(\lambda N)$};
        \node[above=8pt] at ($(rt2)!0.5!(rt3)$) {$\sim\exp(\lambda N)$};
        \node[above=8pt] at ($(rt3)!0.5!(rt4)$) {$\sim\exp(\lambda N)$};
        \node[below=6pt] at (rt3) {$t^*$};
        \node[below=6pt] at (rt2) {$t^*-\delta$};
        \draw[decorate,decoration={brace,amplitude=7pt,mirror}] (bn2) -- (bn3) node[midway, below=8pt,align=center] {Population's previous\\ revision occurs at $t^*-\delta$,\\so $X^N$ is constant over $[t^*-\delta,t^*)$};
        \draw[->] ($(rt3)+(5.7pt,-0.52)$) to [bend left=15] ($(rt3)+(3.5,-1.1)$);
        \node[below=3pt] at ($(rt3)+(4.5,-1.1)$) {\begin{tabular}{c} Revision results\\in a switch from $i$\\to $j$ with probability\\$X_i^N(t^*-\delta)\mathbb{P}(i\rightsquigarrow j \ | \ i)$ \end{tabular}};
    \end{tikzpicture}
    \caption{The revision timeline near $t^*$.}
    \label{fig:str_sw_timeline} 
\end{figure}
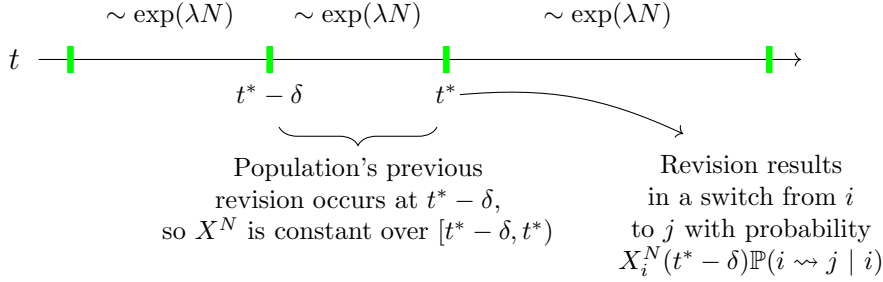

To determine the strategy switching probabilities $\mathbb{P}(i\rightsquigarrow j \ | \ i)$, the agents implement a learning rule, called the revision protocol. This protocol is a Lipschitz continuous function $\mathcal{T}:\Delta\times\mathbb{R}^n\to\mathbb{R}_{\geq 0}^{n\times n}$ that models the agents' strategic preferences in terms of $X^N$ and $P$. Particularly, $\mathcal{T}$ is bounded by the agents' revision rate $\lambda$ and, for any $i,j\in\{1,\dots,n\}$, it specifies the strategy switching probability of an $i$-strategist that is given a revision opportunity at $t^*$ as
\begin{align*}
\mathbb{P}(i\rightsquigarrow j \ | \ i) = \frac{\mathcal{T}_{i,j}(X^N(t^*-\delta),P(t^*-\delta))}{\lambda}.
\end{align*}
Hence, intuitively, $\mathcal{T}_{i,j}$ quantifies the rate with which revising $i$-strategists switch to $j$. Properties of $\mathcal{T}$ are crucial to determine the long run ($t\to\infty$) properties of $X^N$ (e.g. convergence to Nash equilibria). So specific classes of protocols that yield desirable convergence guarantees, such as pairwise comparison and excess payoff \citep{Sandholm2010Population-Game}, has attracted considerable attention. 

\subsection{The Population State as a Markov Process}
\label{subsec:Pop_State_as_Markov_Process}

The strategic environment in Section~\ref{subsec:Strat_Env} and the revision paradigm in Section~\ref{subsec:Rev_Par} act together. The resulting population state $X^N$ is a right-continuous pure-jump Markov process with state space $\mathbb{X}^N$, jump rate $\lambda N$ and transition probabilities from $\mathrm{x}$ to $\mathrm{z}$ given for any $\mathrm{x},\mathrm{z}\in \mathbb{X}^N$ by
\begin{align*}
&\mathbb{P}(X^N(t^*)=\mathrm{z}~|~X^N(t^*-\delta)=\mathrm{x})\\
&=
\begin{cases}
\mathrm{x}_i\mathbb{P}(i\rightsquigarrow j \ | \ i)~\text{if}~\mathrm{z}-\mathrm{x}=\frac{\mathbf{e}^j-\mathbf{e}^i}{N},~i,j\in\{1,\dots,n\}\\
0~\text{otherwise}
\end{cases},
\end{align*}
where $t^*$ denotes a revision time and $t^*-\delta$ is between $t^*$ and the time of the population's previous revision. We refer to \citet[Chapter 10]{Sandholm2010Population-Game} for further details.


\section{Main Results}
\label{sec:Main_Res}

In the population games literature, in most cases, mean-field approximation is the preferred approach for analyzing $X^N$ (see for instance the references at the beginning of Section~\ref{sec:Intro}). This method consists of approximating $X^N$ by the solution $x$ of
\begin{align*}
\dot{x}=\Phi(x)
\end{align*}
where
\begin{align*}
\Phi(\mathrm{x}) := \mathcal{V}(\mathrm{x},\mathcal{F}(\mathrm{x})),
\end{align*}
and, for all $i\in\{1,\dots,n\}$,
\begin{align*}
\mathcal{V}_i(\mathrm{x},\mathrm{p}):=\sum_{j=1}^n \mathrm{x}_j\mathcal{T}_{j,i}(\mathrm{x},\mathrm{p})-\sum_{j=1}^n\mathrm{x}_i\mathcal{T}_{i,j}(\mathrm{x},\mathrm{p}).
\end{align*}
Notice that the Lipschitz continuity of $\mathcal{F}$ and $\mathcal{T}$ extend to $\Phi$, with the Lipschitz constant that we will write as $L_\Phi$. Additionally, we denote $\phi^{max}:=\max_{\mathrm{x}\in\Delta}\Phi(x)$.

Recall from Section~\ref{subsec:Rev_Par} that $\mathcal{T}$ (called the revision protocol) essentially quantifies the strategy switching rates of revising agents. From the proximity of $x_j$ to $X^N_j$, it follows that the term $\sum_{j=1}^n x_i\mathcal{T}_{i,j}(x,p)$ approximates the rate of inflow to strategy $i$ and $\sum_{j=1}^n x_j\mathcal{T}_{j,i}(x,p)$ approximates the rate of outflow from strategy $i$. 

We can interpret the vector field $\Phi$ as an ``approximate derivative.'' Let $\tau^N$ denote the earliest revision time of the population and consider
\begin{align*}
&\Phi^N(\mathrm{x}):= \left(\frac{1}{\mathbb{E}[\tau^N]}\right)\mathbb{E}[X^N(\tau^N)-\mathrm{x} \ | \ X^N(0)=\mathrm{x}]. 
\end{align*}
Observe that $1/\mathbb{E}[\tau^N]=\lambda N$ and
\begin{align*}
&\mathbb{E}[X_i^N(\tau^N)-\mathrm{x}_i \ | \ X^N(0)=\mathrm{x}]\\
&=\frac{1}{N}\Big(\mathbb{P}(\text{First revising agent switches to }i~|~X^N(0)=\mathrm{x})\\
&\qquad\qquad \mathbb{P}(\text{First revising agent switches out from }i~|~X^N(0)=\mathrm{x})\Big)\\
&=\frac{1}{N}\left(\sum_{j=1,j\neq i}^n \mathrm{x}_j\mathbb{P}(j\rightsquigarrow i \ | \ j)-\sum_{j=1,j\neq i}^n \mathrm{x}_i\mathbb{P}(i\rightsquigarrow j \ | \ i)\right)\\
&=\frac{1}{\lambda N}\left(\sum_{j=1,j\neq i}^n \mathrm{x}_j\mathcal{T}_{j,i}(\mathrm{x},\mathcal{F}(\mathrm{x}))-\sum_{j=1,j\neq i}^n \mathrm{x}_i\mathcal{T}_{i,j}(\mathrm{x},\mathcal{F}(\mathrm{x}))\right).
\end{align*}
As a result,
\begin{align*}
&\Phi_i^N(\mathrm{x})= \sum_{j=1,j\neq i}^n \mathrm{x}_j\mathcal{T}_{j,i}(\mathrm{x},\mathcal{F}(\mathrm{x}))-\sum_{j=1,j\neq i}^n \mathrm{x}_i\mathcal{T}_{i,j}(\mathrm{x},\mathcal{F}(\mathrm{x}))= \Phi_i(\mathrm{x}).
\end{align*}

As we briefly expressed in Section~\ref{subsec:Need_for_a_Simpler_Proof}, there are two widespread mean-field approximation results for population games. The first one is by \citet{Sandholm2003Evolution-and-e,Sandholm2010Population-Game}, obtained through a direct application of \citet[Theorem~2.1]{Kurtz1970Solutions-of-Or}.
\begin{mdframed}
\noindent
\textbf{\textit{The Approximation Result of Kurtz \& Sandholm:}} \textit{Let $x$ be the solution of $
\dot{x}=\mathcal{V}(x,\mathcal{F}(x))$ with the initial state $x(0)=\mathrm{x}_0$ for any $\mathrm{x}_0\in\Delta$. If $\lim_{N\to\infty}X^N(0)=\mathrm{x}_0$, then for all $\epsilon, T>0$ it holds that
\begin{align*}
\lim_{N\to\infty} \mathbb{P}\left(\sup_{t\in [0,T]}\|X^N(t)-x(t)\|_2\geq \epsilon\right)=0.
\end{align*}}
\end{mdframed}
The second one, by \citet[Section~6]{Benaim2003Deterministic-a}, establishes a stronger conclusion by revealing a bound on the approximation error for finite $N$.
\begin{mdframed}
\noindent
\textbf{\textit{The Approximation Result of Bena\"{i}m \& Weibull:}} \textit{Let $x$ be the solution of $\dot{x}=\mathcal{V}(x,\mathcal{F}(x))$ with the initial state $x(0)=\mathrm{x}_0$. There exists a scalar $c>0$ such that, for any $\epsilon,T>0$ and large enough $N$:
\begin{align*}
\mathbb{P}\left(\sup_{t\in [0,T]}\|X^N(t)-x(t)\|_{\infty}\geq \epsilon \ \Bigg| \ X^N(0)=\mathrm{x}_0 \right)\leq 2e^{-\epsilon^2cN}
\end{align*}
for all $\mathrm{x}_0\in\Delta$.}
\end{mdframed}

These results are frequently cited and are suitable for many applications, however they have some setbacks.
\begin{itemize}
\item Kurtz's result is much more general than its application to population games. Accordingly, its proof draws results from high-level mathematical subjects, such as martingale theory, infinitesimal generators, and Banach spaces, making it demanding to fully understand. Moreover, it does not offer a bound on the approximation error for finite $N$.
\item Bena\"{i}m and Weibull's result is tailored to population games, and its proof is relatively simpler. However, it still uses results from martingale theory and infinitesimal generators, hindering its accessibility. Additionally, they delegate a significant part of its derivation to the reader, as a reproduction of another proof. The reader obtains explicit expressions for ``large enough $N$'' and $c$ only upon completing this derivation.
\end{itemize}

In light of these comments, we offer the theorem below, which we prove using only undergraduate-level probability.
\begin{mdframed}
\begin{theorem}
\label{thm:main}
Given any $\mathrm{x}_0\in\Delta$, let $x$ be the solution of $\dot{x}=\mathcal{V}(x,\mathcal{F}(x))$ with the initial state $x(0)=\mathrm{x}_0$. For each $\epsilon,T>0$ and $N>\mathfrak{b}$ it holds that
\begin{align}
&\mathbb{P}\left(\sup_{t\in [0,T]}\|X^N(t)-x(t)\|_2\geq \epsilon \ \Bigg| \ X^N(0)=\mathrm{x}_0 \right) \nonumber\\
&\qquad \leq \frac{9\lambda \ell^*}{4N\epsilon^2 L_{\Phi}}(e^{2L_{\Phi}T}-1)+\ell^*\mathcal{H}^N\left(\frac{T}{\ell^*}\right),
\end{align}
where
\begin{align*}
&\mathfrak{b} := \left(\frac{2Te^{L_{\Phi}T}}{\epsilon}\left( L_{\Phi}+\frac{2\phi^{max}\lambda}{L_{\Phi}}(e^{2L_{\Phi}T}-1) \right)\right)^3,\\
&\ell^* := 2\left\lceil\frac{Te\lambda}{\epsilon}\right\rceil,\\
&\mathcal{H}^N(\tau) := \frac{e^{-\lambda N\tau +1}}{2^{\lceil N\epsilon \rceil-1}}.
\end{align*}
\end{theorem}
\end{mdframed}

Notice that $\lim_{N\to\infty}\mathcal{H}^N(T/\ell^*)=0$. Thus, an immediate corollary of Theorem~\ref{thm:main} is that
\begin{align*}
&\lim_{N\to\infty}\mathbb{P}\left(\sup_{t\in [0,T]}\|X^N(t)-x(t)\|_2\geq \epsilon \ \Bigg| \ X^N(0)=\mathrm{x}_0 \right)=0.
\end{align*}

\subsection{Proof of Theorem~\ref{thm:main}}

Fundamentally, we prove Theorem~\ref{thm:main} by separately analyzing the deviations of $X^N$ and $x$ from the expected trajectory $\bar{X}^N := \mathbb{E}[X^N\ | \ X^N(0)=\mathrm{x}_0]$: From $\|X^N(t)-x(t)\|_2\leq \|X^N(t)-\bar{X}^N(t)\|_2 + \|\bar{X}^N(t)-x(t)\|_2$, it follows that
\begin{align*}
\sup_{t\in[0,T]}\|X^N(t)-&\bar{X}^N(t)\|_2 < \frac{\epsilon}{2} \quad \text{and} \quad  \sup_{t\in[0,T]}\|\bar{X}^N(t)-x(t)\|_2 < \frac{\epsilon}{2}\\
&\Rightarrow \sup_{t\in[0,T]}\|X^N(t)-x(t)\|_2<\epsilon.
\end{align*}
Hence, to prove Theorem~\ref{thm:main}, it suffices to verify that
\begin{align*}
&\mathbb{P}\Bigg( \sup_{t\in[0,T]}\|X^N(t)-\bar{X}^N(t)\|_2 \geq \frac{\epsilon}{2} \text{ or } \sup_{t\in[0,T]}\|\bar{X}^N(t)-x(t)\|_2 \geq \frac{\epsilon}{2} \ \Bigg| \ X^N(0)=\mathrm{x}_0 \Bigg) \\
&\leq \frac{9\lambda \ell^*}{4N\epsilon^2 L_{\Phi}}(e^{2L_{\Phi}T}-1)+\ell^*\mathcal{H}^N\left(\frac{T}{\ell^*}\right)
\end{align*}
whenever $N>\mathfrak{b}$.

Furthermore, observe that $\bar{X}^N$ and $x$ are deterministic. Therefore, our aim shifts into showing that if $N>\mathfrak{b}$, then
\begin{align*}
& \sup_{t\in[0,T]}\|\bar{X}^N(t)-x(t)\|_2 < \frac{\epsilon}{2}
\end{align*}
and
\begin{align}
&\mathbb{P}\Bigg( \sup_{t\in[0,T]}\|X^N(t)-\bar{X}^N(t)\|_2 \geq \frac{\epsilon}{2}\ \Bigg| \ X^N(0)=\mathrm{x}_0 \Bigg) \nonumber\\
&\leq \frac{9\lambda \ell^*}{4N\epsilon^2 L_{\Phi}}(e^{2L_{\Phi}T}-1)+\ell^*\mathcal{H}^N\left(\frac{T}{\ell^*}\right). \label{eq:X_Xbar_bd}
\end{align}

In fact, \eqref{eq:X_Xbar_bd} holds for all $N$, regardless of whether it exceeds $\mathfrak{b}$. We will establish this bound in Lemma~\ref{lem:X_Xbar_Bd}. Subsequently, in Lemma~\ref{lem:Xbar_x_Bd}, we will prove that $N>\mathfrak{b}$ implies $\sup_{t\in[0,T]}\|\bar{X}^N(t)-x(t)\|_2 < \frac{\epsilon}{2}$.

We would like to highlight that a crucial step in obtaining \eqref{eq:X_Xbar_bd} is bounding the variance of $X^N$, which we state in Lemma~\ref{lem:Var_bd}. Specifically, we derive \eqref{eq:X_Xbar_bd} from the following steps:
\begin{itemize}
\item We bound the variance of $X^N$ and quantize $[0,T]$ as $$\mathbb{T}:=\left\{0,\frac{T}{\ell^*},\dots,\frac{(\ell^*-1)T}{\ell^*}\right\}.$$ Then, we use the variance bound on $X^N(t)$ for every $t\in\mathbb{T}$. This causes the $(e^{2L_{\Phi}T}-1)9\lambda\ell^*/(4N\epsilon^2 L_{\Phi})$ term.
\item We show that, with high probability, $X^N(t)$ can't deviate significantly from $\bar{X}^N(t)$ over $t\in [kT/\ell^*,(k+1)T/\ell^*]$ for every $k\in\{0,\dots,l^*-1\}$. This results in the $\ell^*\mathcal{H}^N(T/\ell^*)$ term.
\end{itemize}

Overall, below are some key results that we leverage in the remainder of the proof.
\begin{itemize}
\item Law of total variance.
\item Law of total probability.
\item Cauchy-Schwarz inequality for random variables (also known as the covariance inequality).
\item Gr\"{o}nwall-Bellman Lemma (also known as Gr\"{o}nwall's inequality).
\item Markov's inequality.
\end{itemize}

\subsubsection{Deviation of $X^N$ from $\bar{X}^N$}

In this section, we prove \eqref{eq:X_Xbar_bd}. We begin our analysis by presenting a bound for the variance of $X^N$.

\begin{lemma}
\label{lem:Var_bd}
For every $N\in\mathbb{N}$, $\mathrm{x}_0\in\mathbb{X}^N$ and $t\geq 0$, it holds that
\begin{align*}
\sigma^N(t):= Var[X^N(t)\ | \ X^N(0)=\mathrm{x}_0]\leq \left(\frac{2\lambda}{N}\right)te^{2L_{\Phi}t}.
\end{align*}
\end{lemma}

\begin{proof}
For the entire proof, we assume that $X^N(0)=\mathrm{x}_0$, but do not explicitly denote this condition for notational convenience. Let us start by using the law of total variance to write
\begin{equation}
\sigma^N(t+\delta) = \mathbb{E} \Big [Var \big [X^N(t+\delta)~|~X^N(t) \big ] \Big ] + Var \Big [\mathbb{E} \big [X^N(t+\delta)~|~X^N(t) \big] \Big ]. \label{eq:tot_var}
\end{equation}

Recall from the properties of the Poisson process (see Section~\ref{subsubsec:Rev_Times}) that
\begin{align*}
\lim_{\delta\to 0^+}\frac{\mathbb{P}(\mathbb{I}([t,t+\delta])=1)}{\delta}= N\lambda
\end{align*}
and
\begin{align*}
\lim_{\delta\to 0^+}\frac{\mathbb{P}(\mathbb{I}([t,t+\delta])\geq 2)}{\delta^2}= \lim_{\delta\to 0^+}\frac{\mathcal{M}^N(\delta)}{\delta^2} = \frac{(N\lambda)^2}{2}.
\end{align*}

With these in mind, we focus on the $Var\big [X^N(t+\delta)~|~X^N(t) \big]$ term. We can write
\begin{align*}
&Var \big [X^N(t+\delta)~|~X^N(t) \big ]\\
&= Var \big [X^N(t+\delta)-X^N(t)~|~X^N(t),~\mathbb{I}([t,t+\delta]=1)\big ] (N \lambda \delta) e^{- N \lambda \delta} \\ 
&\quad  + Var \big [X^N(t+\delta)-X^N(t)~|~X^N(t),~\mathbb{I}([t,t+\delta]\geq 2)\big ]\mathcal{M}^N(\delta)
\end{align*}
and use
\begin{align*}
&Var \big [X^N(t+\delta)-X^N(t)~|~X^N(t),~\mathbb{I}([t,t+\delta]=1)\big ] \leq \frac{\|\mathbf{e}^i-\mathbf{e}^j\|^2_2}{N^2}=\frac{2}{N^2},\\
&Var \big [X^N(t+\delta)-X^N(t)~|~X^N(t),~\mathbb{I}([t,t+\delta]\geq 2)\big ] \leq \max_{\mathrm{x},\mathrm{y}\in\Delta}\|\mathrm{x}-\mathrm{y}\|_2^2=4
\end{align*}
to get
\begin{align}
&Var \big [X^N(t+\delta)-X^N(t)~|~X^N(t)\big ] \leq \frac{2\lambda\delta}{N}e^{-N\lambda\delta} + 4\mathcal{M}^N(\delta). \label{eq:var_diff_bd}
\end{align}

Now, we examine $\mathbb{E} \big [X^N(t+\delta)~|~X^N(t) \big ]$. Looking at the effect of the $\delta$ time increment, we get
\begin{align*}
&\mathbb{E} \big [X^N(t+\delta)~|~X^N(t)=\mathrm{x} \big ]\\
&= \mathrm{x} + \underbrace{\left ( \frac{1}{N} \sum_{i,j=1}^n (\mathbf{e}^j-\mathbf{e}^i) \mathrm{x}_i \mathbb{P}(i \rightsquigarrow j \ | \ i) \right) }_{\text{effect of one revision}} (N \lambda \delta) e^{- N \lambda \delta} + \underbrace{\Gamma(\mathrm{x},\delta)}_{\text{effect of $\geq 2$ rev}} \mathcal{M}^N(\delta)\\
&=\mathrm{x} + \delta \Phi(\mathrm{x})e^{-N\lambda\delta} + \Gamma(\mathrm{x},\delta) \mathcal{M}^N(\delta),
\end{align*}
where
\begin{align*}
\|\Gamma(\mathrm{x},\delta)\|_2\leq \max_{\mathrm{x},\mathrm{y}\in\Delta}\|\mathrm{x}-\mathrm{y}\|_2=2.
\end{align*}

Consequently, employing the properties of variance, the covariance inequality, and our analyses of $\mathbb{E} \big [X^N(t+\delta)~|~X^N(t) \big ]$ and $Var \big [X^N(t+\delta)~|~X^N(t) \big ]$, we obtain
\begin{align}
&Var \Big [ \mathbb{E} \big [X^N(t+\delta)~|~X^N(t) \big ] \Big] \nonumber\\
& \leq Var[Z(t+\delta)] + 2 \sqrt{Var[Z(t+\delta)]}\sqrt{Var[W(t+\delta)]} + Var[W(t+\delta)] \label{eq:Var_of_exp}
\end{align} 
where 
\begin{align*} 
Z(t+\delta)&:=X^N(t)+\delta \Phi(X^N(t))e^{-N\lambda\delta}, \\ 
W(t+\delta)&:= \Gamma(X^N(t),\delta)\mathcal{M}^N(\delta).
\end{align*}
Similarly, leveraging the properties of variance and the covariance inequality, we can write
\begin{align}
&Var[Z(t+\delta)] \nonumber\\
&\leq \sigma^N(t) + 2\delta e^{-N\lambda\delta} \sqrt{\sigma^N(t)} \sqrt{Var[\Phi(X^N(t))]}+\delta^2e^{-2N\lambda\delta}Var[\Phi(X^N(t))] \nonumber\\ 
& \leq \sigma^N(t) + 2\delta e^{-N\lambda\delta} L_\Phi \sigma^N(t) + \delta^2e^{-2N\lambda\delta}L_\Phi^2 \sigma^N(t). \label{eq:Var_Z_bd}
\end{align}
In the last inequality, note that we used $Var[\Phi(X^N(t))]\leq L_{\Phi}^2\sigma^N(t)$, which follows from
\begin{align*}
2Var[\Phi(X^N(t))] &= Var[\Phi(X^N(t))-\Phi(Y^N(t))]\\
&= \mathbb{E}[\|\Phi(X^N(t))-\Phi(Y^N(t))\|_2^2]-\|\mathbb{E}[\Phi(X^N(t))-\Phi(Y^N(t))]\|_2^2\\
&= \mathbb{E}[\|\Phi(X^N(t))-\Phi(Y^N(t))\|_2^2]\\
&\leq L_{\Phi}^2\mathbb{E}[\|X^N(t)-Y^N(t)\|_2^2]\\
&= L_{\Phi}^2\mathbb{E}[\|X^N(t)-Y^N(t)\|_2^2]-\|\mathbb{E}[X^N(t)-Y^N(t)]\|_2^2\\
&= 2L_{\Phi}^2\sigma^N(t),
\end{align*}
where $Y^N(t)$ is a random variable that is independent to and identically distributed as $X^N(t)$. 

As a result, substituting \eqref{eq:var_diff_bd}, \eqref{eq:Var_of_exp}, \eqref{eq:Var_Z_bd} and $\|\Gamma(X^N(t),\delta)\|_2\leq 2$ into \eqref{eq:tot_var} yields
\begin{align*}
\sigma^N(t+\delta) \leq \sigma^N(t) + \frac{2 \lambda \delta}{N} e^{-N\lambda\delta} +2 \delta e^{-N\lambda\delta}L_\Phi\sigma^N(t)+O(\delta^2),
\end{align*}
in which $O(\delta^2)$ represents the sum of the terms $g(\delta)$ that satisfy 
\begin{align*}
\lim_{\delta\to 0^+}\frac{g(\delta)}{\delta^2}=a
\end{align*}
for some constant $a$. Therefore,
\begin{align*}
&\lim_{\delta\to 0^+}\frac{\sigma^N(t+\delta)-\sigma^N(t)}{\delta} \leq \lim_{\delta\to 0^+} \frac{2 \lambda}{N} e^{-N\lambda\delta} +2 e^{-N\lambda\delta}L_\Phi\sigma^N(t)+\frac{O(\delta^2)}{\delta}\\
&\Rightarrow \dot{\sigma}^N(t) \leq \frac{2\lambda}{N}+2L_{\Phi}\sigma^N(t).
\end{align*}
Finally, noticing that $\sigma^N$ is continuous, we can use the Bellman-Gr\"{o}nwall Lemma to conclude
\begin{align*}
&\sigma^N(t)\leq \frac{2\lambda}{N} te^{2L_{\Phi}t}.
\end{align*}

$\blacksquare$
\end{proof}

Having derived a bound on the variance of $X^N$, we proceed to use it for bounding $\mathbb{P}(\sup_{t\in[0,T]}\|X^N(t)-\bar{X}^N(t)\|_2>\epsilon/2~|~ X^N(0)=\mathrm{x}_0)$. As we pointed out in Section~\ref{sec:Main_Res}, our strategy is to first quantize $[0,T]$ as $$\mathbb{T}:=\left\{0,\frac{T}{\ell^*},\dots,\frac{(\ell^*-1)T}{\ell^*}\right\}$$ and bound $\mathbb{P}(\max_{t\in\mathbb{T}}\|X^N(t)-\bar{X}^N(t)\|_2>\epsilon/2~|~ X^N(0)=\mathrm{x}_0)$. Subsequently, we will prove that $X^N$ can't deviate from $\bar{X}^N$ with high probability over any two consecutive points of $\mathbb{T}$.

\begin{lemma}
\label{lem:X_Xbar_Bd}
For every $\epsilon,T>0$, $N\in\mathbb{N}$ and $\mathrm{x}_0\in\Delta$, the following holds:
\begin{align*}
&\mathbb{P}\Bigg( \sup_{t\in[0,T]}\|X^N(t)-\bar{X}^N(t)\|_2 \geq \frac{\epsilon}{2}\ \Bigg| \ X^N(0)=\mathrm{x}_0 \Bigg) \nonumber\\
&\leq \frac{9\lambda \ell^*}{4N\epsilon^2 L_{\Phi}}(e^{2L_{\Phi}T}-1)+\ell^*\mathcal{H}^N\left(\frac{T}{\ell^*}\right).
\end{align*}
\end{lemma}

\begin{proof}
Observe that, if 
\begin{align*}
& \mathbb{I}\left(\left[\frac{kT}{\ell^*},\frac{(k+1)T}{\ell^*}\right]\right)<\left\lceil N\epsilon \right\rceil,
\end{align*}
then $\|X^N(t)-X^N(s)\|_2\leq 2\epsilon$ for every $t,s\in [kT/\ell^*,(k+1)T/\ell^*]$. In other words, if the population receives less than $\lceil N\epsilon\rceil$ revision opportunities within a time interval of length $T/\ell^*$, then $X^N$ cannot deviate more than $2\epsilon$ from $\bar{X}^N$ over this interval. If, in addition 
\begin{align*}
\left\|X^N\left(\frac{(k+1)T}{\ell^*}\right)-\bar{X}^N\left(\frac{(k+1)T}{\ell^*}\right)\right\|_2<\epsilon
\end{align*}
for all $k\in\{0,\dots,\ell^*-1\}$, then
\begin{align*}
\sup_{t\in[0,T]} \|X^N(t)-\hat{X}^N(t)\|_2 <3\epsilon,
\end{align*}
where $\hat{X}^N$ is the staircase approximation of $\bar{X}^N$ given by $\hat{X}^N(t):=\bar{X}^N(\lceil t\ell^*/T \rceil T/\ell^*)$. 

From this reasoning, we get the union-bound below:
\begin{align*}
&\mathbb{P}\left( \sup_{t\in[0,T]}\|X^N(t)-\hat{X}^N(t)\|_2\geq 3\epsilon ~\Bigg|~ X^N(0)=\mathrm{x}_0\right)\\
&\leq \mathbb{P}\Bigg(\max_{t\in \mathbb{T}}\|X^N(t)-\bar{X}^N(t)\|_2\geq \epsilon\\
&\qquad\qquad \bigcup_{k=0}^{\ell^*-1} ~ \mathbb{I}\left(\left[\frac{kT}{\ell^*},\frac{(k+1)T}{\ell^*}\right]\right)\geq \lceil N\epsilon \rceil \ \Bigg|\ X^N(0)=\mathrm{x}_0 \Bigg)\\
&\leq \mathbb{P}\Bigg( \bigcup_{k=0}^{\ell^*-1} \Bigg\{ \left\|X^N\left(\frac{kT}{\ell^*}\right)-\bar{X}^N\left(\frac{kT}{\ell^*}\right)\right\|_2\geq \epsilon\\
&\qquad\qquad \cup~ \mathbb{I}\left(\left[\frac{kT}{\ell^*},\frac{(k+1)T}{\ell^*}\right]\right)\geq \lceil N\epsilon \rceil \Bigg\} \ \Bigg|\ X^N(0)=\mathrm{x}_0 \Bigg)\\
&\leq \sum_{k=0}^{\ell^*-1} \mathbb{P}\left( \left\|X^N\left(\frac{kT}{\ell^*}\right)-\bar{X}^N\left(\frac{kT}{\ell^*}\right)\right\|_2 \geq \epsilon ~\Bigg|~ X^N(0)=\mathrm{x}_0 \right)\\
&\qquad\qquad +\sum_{k=0}^{\ell^*-1}\mathbb{P}\left( \mathbb{I}\left(\left[\frac{kT}{\ell^*},\frac{(k+1)T}{\ell^*}\right]\right)\geq \lceil N\epsilon \rceil ~\Bigg|~ X^N(0)=\mathrm{x}_0 \right).
\end{align*}

Let us focus on $\mathbb{P}(\|X^N(kT/\ell^*)-\bar{X}^N(kT/\ell^*)\|_2\geq \epsilon ~|~ X^N(0)=\mathrm{x}_0 )$. Notice that, for any $t\in\mathbb{T}$, we can use Markov's inequality to obtain
\begin{align*}
\mathbb{P}\left( \|X^N(t)-\bar{X}^N(t)\|_2^2 \geq \epsilon^2 ~|~ X^N(0)=\mathrm{x}_0 \right) \leq \frac{\sigma^N(t)}{\epsilon^2}.
\end{align*}
Thus,
\begin{align*}
\sum_{k=0}^{\ell^*-1} \mathbb{P}\left( \left\|X^N\left(\frac{kT}{\ell^*}\right)-\bar{X}^N\left(\frac{kT}{\ell^*}\right)\right\|_2 \geq \epsilon ~\Bigg|~ X^N(0)=\mathrm{x}_0 \right) \leq \frac{1}{\epsilon^2}\sum_{k=0}^{\ell^*-1} \sigma^N\left(\frac{kT}{\ell^*}\right).
\end{align*}
From Lemma~\ref{lem:Var_bd}, we know that
\begin{align*}
\sigma^N\left(\frac{kT}{\ell^*}\right)\leq \left(\frac{2\lambda}{N}\right)\left(\frac{kT}{\ell^*}\right)e^{2L_{\phi}kT/\ell^*}.
\end{align*}
Moreover, since $te^{L_{\Phi}t}$ is increasing in $t$, the sum
\begin{align*}
&\sum_{k=0}^{\ell^*-1}\left(\frac{kT}{\ell^*}\right)e^{kT/\ell^*}=\frac{\ell^*}{T}\sum_{k=0}^{\ell^*-1}\left(\frac{kT}{\ell^*}\right)e^{kT/\ell^*}\frac{T}{\ell^*}(k+1-k)
\end{align*}
is a lower bound for the integral $(\ell^*/T)\int_0^T te^{L_{\Phi}t}dt$. Therefore,
\begin{align*}
&\sum_{k=0}^{\ell^*-1} \mathbb{P}\left( \left\|X^N\left(\frac{kT}{\ell^*}\right)-\bar{X}^N\left(\frac{kT}{\ell^*}\right)\right\|_2 \geq \epsilon ~\Bigg|~ X^N(0)=\mathrm{x}_0 \right)\\
&\leq \frac{\lambda \ell^*}{N L_{\Phi} \epsilon^2}(e^{2L_{\Phi}T}-1).
\end{align*}

Now, we consider $\mathbb{P}(\mathbb{I}([kT/\ell^*,(k+1)T/\ell^*])>\lceil N\epsilon \rceil ~|~ X^N(0)=\mathrm{x}_0)$. Assume that $\tau$ satisfies $2\tau\leq\beta:=\epsilon/(e\lambda)$ and note that $\mathbb{I}([t,t+\tau])$ has Poisson distribution with mean $N\lambda\tau$. To get a tail bound, we will use the Chernoff bound on $\mathbb{I}([t,t+\tau])$ in the form
\begin{align*}
\mathbb{P}(\mathbb{I}([t,t+\tau])\geq (1+\eta)\mu)\leq \left(\frac{e^\eta}{(1+\eta)^{1+\eta}} \right)^\mu.
\end{align*}
Specifically, we set $\mu=N\lambda\tau$ and $\eta=\lceil N\epsilon \rceil /(N\lambda\tau)-1$ to get
\begin{align*}
\mathbb{P}(\mathbb{I}([t,t+\tau])\geq \lceil N\epsilon \rceil )\leq e^{\lceil N\epsilon \rceil} e^{-N\lambda\tau} \left(\frac{\lambda N\tau}{\lceil N\epsilon \rceil}\right)^{\lceil N\epsilon \rceil}.
\end{align*}
From $\tau\leq \beta/2$, it follows that
\begin{align*}
e^{\lceil N\epsilon \rceil} e^{-N\lambda\tau} \left(\frac{\lambda N\tau}{\lceil N\epsilon \rceil}\right)^{\lceil N\epsilon \rceil}\leq e^{\lceil N\epsilon \rceil} e^{-N\lambda\tau} \left(\frac{1}{2e}\right)^{ N\epsilon}\left(\frac{N\epsilon}{\lceil N\epsilon \rceil}\right)\leq e^{\lceil N\epsilon \rceil} e^{-N\lambda\tau} \left(\frac{1}{2e}\right)^{ N\epsilon}.
\end{align*}
Moreover, $N\epsilon +1\geq \lceil N\epsilon \rceil$ yields
\begin{align*}
&e^{\lceil N\epsilon \rceil} e^{-N\lambda\tau} \left(\frac{1}{2e}\right)^{ N\epsilon}\\
&\leq e^{\lceil N\epsilon \rceil} e^{-N\lambda\tau} \left(\frac{1}{2e}\right)^{ \lceil N\epsilon\rceil -1}\\
&= \frac{e^{-N\lambda\tau +1}}{2^{\lceil N\epsilon \rceil -1}}\\
&=:\mathcal{H}^N(\tau).
\end{align*}
Recall that $\ell^*:=2\lceil Te\lambda/\epsilon \rceil$, which implies $2T/\ell^*\leq\beta$. Thus, we can use the $\mathcal{H}^N(\tau)$ bound with $\tau = T/\ell^*$ to arrive at
\begin{align*}
&\sum_{k=0}^{\ell^*-1}\mathbb{P}\left( \mathbb{I}\left(\left[\frac{kT}{\ell^*},\frac{(k+1)T}{\ell^*}\right]\right)\geq \lceil N\epsilon \rceil ~\Bigg|~ X^N(0)=\mathrm{x}_0 \right)\leq \ell^* \mathcal{H}^N\left(\frac{T}{\ell^*} \right).
\end{align*}
Note that there are sharper tail bounds, however the one above is sufficient for our purposes because $\lim_{N\to \infty}\ell^*\mathcal{H}^N(T/\ell^*)=0$.

Up to this point, we have verified that
\begin{align*}
&\mathbb{P}\left( \sup_{t\in[0,T]}\|X^N(t)-\hat{X}^N(t)\|_2\geq 3\epsilon ~\Bigg|~ X^N(0)=\mathrm{x}_0\right)\\
&\leq \frac{\lambda \ell^*}{N L_{\Phi} \epsilon^2}(e^{2L_{\Phi}T}-1) +  \ell^* \mathcal{H}^N\left(\frac{T}{\ell*} \right).
\end{align*}
The conclusion of the lemma follows from the continuity of $\bar{X}^N$ and by replacing $3\epsilon$ with $\epsilon/2$.

$\blacksquare$
\end{proof}

\subsubsection{Deviation of $\bar{X}^N$ from $x$}

We conclude the proof of Theorem~\ref{thm:main} by establishing that $N>\mathfrak{b}$ implies $\sup_{t\in[0,T]}\|\bar{X}^N(t)-x(t)\|_2< \epsilon/2$.

\begin{lemma}
\label{lem:Xbar_x_Bd}
Given any $\epsilon,T>0$, if
\begin{align*}
N>\mathfrak{b}:=\left(\frac{2Te^{L_{\Phi}T}}{\epsilon}\left( L_{\Phi}+\frac{2\phi^{max}\lambda}{L_{\Phi}}(e^{2L_{\Phi}T}-1) \right)\right)^3,
\end{align*}
then
\begin{align*}
\sup_{t\in[0,T]}\|\bar{X}^N(t)-x(t)\|_2< \frac{\epsilon}{2}.
\end{align*}
\end{lemma}

\begin{proof}
Let us define $e(t):= \bar{X}^N(t)-x(t)$. Recall from the proof of Lemma~\ref{lem:Var_bd} that
\begin{align*}
\mathbb{E} \big [X^N(t+\delta) ~|~X^N(t) = \mathrm{x} \big ] = \mathrm{x} + \delta \Phi(\mathrm{x})e^{-N\lambda\delta} + \Gamma(\mathrm{x},\delta) \mathcal{M}^N(\delta).
\end{align*}
Therefore,
\begin{align*}
&\frac{\mathbb{E} \big [X^N(t+\delta) ~|~X^N(t)=\mathrm{x} \big ]-\mathrm{x}}{\delta} = \Phi(\mathrm{x})e^{-N\lambda\delta} + \frac{\Gamma(\mathrm{x},\delta) \mathcal{M}^N(\delta)}{\delta}\\
&\Rightarrow \lim_{\delta\to 0^+} \frac{\mathbb{E} \big [X^N(t+\delta) ~|~X^N(t)=\mathrm{x} \big ]-\mathrm{x}}{\delta} = \Phi(\mathrm{x})\\
&\Rightarrow \bar{X}^N(t) = \mathrm{x}_0 + \int_0^t \mathbb{E} \left[ \Phi(X^N(\tau))~|~X^N(0)=\mathrm{x}_0 \right ] d \tau.
\end{align*}

For any $t\in[0,T]$ and $\upsilon\geq 0$, let us define $\mathfrak{E}(t)$ as the event that $\|X^N(t)-\bar{X}^N(t)\|_2\leq v$ (conditioned on $X^N(0)=\mathrm{x}_0$) and denote $\varrho^N(\upsilon,t):= \mathbb{P}(\mathfrak{E}^c(t))$. Then, we can write
\begin{align*}
&\bar{X}^N(t)=  x(0) + \int_0^t \Phi(\bar{X}^N(\tau)) + u(\tau) d \tau,
\end{align*}
where 
\begin{align*} 
u(t)&:=\mathbb{E} \left[ \Phi(X^N(t)) - \Phi(\bar{X}^N(t))~|~\mathfrak{E}(t) \right ](1-\varrho^N(\upsilon,t))\\ 
&\qquad +\mathbb{E} \left[ \Phi(X^N(t)) - \Phi(\bar{X}^N(t))~|~\mathfrak{E}^c(t) \right ]\varrho^N(\upsilon,t).
\end{align*}
Observe that 
\begin{align*}
\varrho^N(\upsilon,t) \leq  \frac{\sigma^N(t)}{\upsilon^2} \leq \frac{2 \lambda}{N \upsilon^2}  t e^{2 L_\Phi t},
\end{align*}
which yields
\begin{align*}
\| u(t) \|_2 \leq \upsilon L_\Phi + 2 \phi^{max} \varrho^N(\upsilon,t).
\end{align*}

Consequently, from
\begin{align*}
&e(t)=\bar{X}^N(t) - x(t)=\int_0^t \Phi(\bar{X}^N(\tau))-\Phi(x(\tau)) + u(\tau) d \tau,
\end{align*}
it follows that
\begin{align*} 
&\| e(t) \|_2\\
&\leq \int_0^t \|u(\tau)\|_2 + \| \Phi(\bar{X}^N(\tau)) - \Phi(x(\tau)) \|_2 d \tau\\
&\leq \int_0^t \|u(\tau)\|_2 + L_\Phi \| e(\tau) \|_2 d \tau.
\end{align*}
Using the Bellman-Gr\"{o}nwall Lemma on this expression gives 
\begin{equation*}
\| e(t) \|_2 \leq e^{L_\Phi t} \int_0^t \|u(\tau)\|_2 d \tau \leq \left ( \upsilon L_\Phi  + \frac{2\phi^{max} \lambda}{N \upsilon^2 L_\Phi}(e^{2 L_\Phi t}-1) \right ) t e^{L_\Phi t},
\end{equation*}
implying that
\begin{equation*}
\sup_{t\in[0,T]}\| e(t) \|_2 \leq \left ( \upsilon L_\Phi  + \frac{2\phi^{max} \lambda}{N \upsilon^2 L_\Phi}(e^{2 L_\Phi T}-1) \right ) T e^{L_\Phi T}.
\end{equation*}

Since $\upsilon\geq 0$ was arbitrary, we can set it as $N^{-1/3}$. Hence, 
\begin{equation*}
\sup_{t\in[0,T]}\| e(t) \|_2 \leq N^{-1/3}\left ( L_\Phi  + \frac{2\phi^{max} \lambda}{L_\Phi}(e^{2 L_\Phi T}-1) \right ) T e^{L_\Phi T}.
\end{equation*}
The lemma follows from bounding the right hand side of this expression by $\epsilon/2$ and solving for $N$.

$\blacksquare$
\end{proof}

\bibliographystyle{plainnat}
\bibliography{KaraRefs,MartinsRefs}

\end{document}